\documentclass[10pt, conference,compsocconf]{IEEEtran}
\IEEEoverridecommandlockouts
\usepackage[utf8]{inputenc}
\usepackage[T1]{fontenc}
\usepackage[english]{babel}
\usepackage{ntheorem}
\usepackage{amsmath,amssymb}
\usepackage{epsfig,psfrag}
\usepackage{graphics,graphicx} 
\usepackage{color}
\usepackage{dsfont}
\usepackage{graphicx}
\usepackage{subfigure}
\usepackage{stmaryrd}
\usepackage{color}
\usepackage{cite}

\begin{document}

\newtheorem{corollary}{Corollary}
\newtheorem{proposition}{Proposition}
\newtheorem{Conjec}{Conjecture}
\newtheorem{definition}{Definition}
\newtheorem{theorem}{Theorem}

\newtheorem{notation}{Notation}



\newcommand{\rel}[0]{\ensuremath{{\mathcal{R}}}}
\newcommand{\Gall}[0]{\ensuremath{\mathcal{G}}}

\newcommand{\harvey}[0]{\ensuremath{\textsf{haRVey}\xspace}}
\newcommand{\barvey}[0]{\ensuremath{\textsf{barvey}\xspace}}
\newcommand{\bamTorv}[0]{\ensuremath{\textsf{bam2rv}\xspace}}
\newcommand{\rvqe}[0]{\ensuremath{\textsf{rvqe}\xspace}}
\newcommand{\Why}[0]{\ensuremath{\textsf{Why}\xspace}}
\newcommand{\Be}[0]{\textsf{B}\xspace}
\newcommand{\Zed}[0]{\ensuremath{\textsf{Z}\xspace}}

\newcommand{\Lang}[0]{\ensuremath{{\cal L}}}
\newcommand{\Sorts}[0]{\ensuremath{{\cal S}}}
\newcommand{\Vars}[0]{\ensuremath{{\cal V}}}
\newcommand{\Funcs}[0]{\ensuremath{{\cal F}}}
\newcommand{\Preds}[0]{\ensuremath{{\cal P}}}
\newcommand{\arite}[0]{\ensuremath{{\alpha}}} 
\newcommand{\sorte}[0]{\ensuremath{{\sigma}}}
\newcommand{\tauTerm}[0]{\textit{$\tau$-terme}\xspace}
\newcommand{\tauITerm}[1]{\textit{$\tau_{#1}$-term}\xspace}
\newcommand{\atome}[0]{\textit{atome}\xspace}
\newcommand{\fpoms}[0]{\textit{fpoms}\xspace}
\newcommand{\VLibre}[1]{\Vars_{\textit{lib}}(#1)}
\newcommand{\congruenceE}[0]{\ensuremath{{\cal E}}}
\newcommand{\preInterpret}[1]{\ensuremath{[ #1 ]}}
\newcommand{\preInterpretSSET}[1]{\preInterpret{#1}}
\newcommand{\preInterpretBAES}[1]{I'(#1)}
\newcommand{\preinterpret}[0]{\preInterpret{\,}}
\newcommand{\valuationSSET}[1]{\valuation(#1)}
\newcommand{\valuationBAES}[1]{\valuation'(#1)}
\newcommand{\valuationX}[1]{\valuation(#1)}
\newcommand{\valuation}[0]{\ensuremath{\rho}\xspace}
\newcommand{\etat}[0]{\ensuremath{\rho}}
\newcommand{\constt}[1]{\ensuremath{\textsf{#1}}}
\newcommand{\termSet}[0]{t}

\newcommand{\sset}[0]{\ensuremath{\mathit{SSET}}}
\newcommand{\axCod}[0]{\ensuremath{\delta}}
\newcommand{\mty}[0]{\ensuremath{\mathsf{mty}}}
\newcommand{\ttrue}[0]{\ensuremath{\mathsf{tt}}}
\newcommand{\tfalse}[0]{\ensuremath{\mathsf{ff}}}

\newcommand{\interpretI}[0]{\ensuremath{{\cal I}}\xspace}
\newcommand{\classInterpret}[0]{\ensuremath{\mathbf{I}\xspace}}
\newcommand{\ltrue}[0]{\ensuremath{\top}}
\newcommand{\lfalse}[0]{\ensuremath{\bot}}
\newcommand{\modele}[0]{\ensuremath{\Vdash}}
\newcommand{\deduc}[0]{\vdash}

\newcommand{\conseq}[0]{\ensuremath{\vDash}}
\newcommand{\herbrandU}[0]{\ensuremath{\mathbbm{H}}\xspace}

\newcommand{\Nats}[0]{\ensuremath{\mathbb{N}}}
\newcommand{\Z}[0]{\ensuremath{\mathbb{Z}}}
\newcommand{\R}[0]{\ensuremath{\mathbb{R}}}
\newcommand{\Bool}[0]{\ensuremath{\mathds{B}}}
\newcommand{\StratSet}[0]{\ensuremath{\mathbb{S}}}

\newcommand{\impliq}[0]{\ensuremath{{\Rightarrow}}}
\newcommand{\equival}[0]{\ensuremath{{\Leftrightarrow}}}

\newcommand{\sortedForall}[3]{\ensuremath{(\forall_{#1} #2\, . \, #3)}}
\newcommand{\sortedExists}[3]{\ensuremath{(\exists_{#1} #2\, . \, #3)}}
\newcommand{\sortedQ}[3]{\ensuremath{(Q_{#1} #2\, . \, #3)}}
\newcommand{\Forall}[2]{\ensuremath{(\forall #1\, . \, #2)}}
\newcommand{\qnnf}[0]{\ensuremath{_{\textit{nf}}}}
\newcommand{\gq}{\ensuremath{\mathsf{gq}}}
\newcommand{\gqe}{\ensuremath{\mathsf{gqe}}}
\newcommand{\bdd}[1]{\textit{bdd}(#1)}

\newcommand{\regleSep}[0]{\mathtt{~\textbf{|}~}}
\newcommand{\regleDeriv}[0]{\mathtt{::=}}

\newcommand{\transfo}[0]{\ensuremath{{\rightarrow}}}
\newcommand{\negativeForm}[1]{\ensuremath{\textit{neg}({#1})}}
\newcommand{\cnf}[1]{\ensuremath{\textit{cnf}({#1})}}
\newcommand{\dnf}[1]{\ensuremath{\textit{dnf}({#1})}}
\newcommand{\polar}[1]{\ensuremath{\textit{pol}(#1)}}

\newcommand{\statesVars}[0]{\ensuremath{X}}
\newcommand{\statesSet}[0]{\ensuremath{{\Sigma}}}
\newcommand{\initialAssertion}[0]{\ensuremath{I}}
\newcommand{\transSymbol}[0]{\ensuremath{{\rightarrow}}}
\newcommand{\transLabel}[0]{\ensuremath{{\textit{act}}}}
\newcommand{\prodASync}[0]{\ensuremath{\otimes}}
\newcommand{\prodSync}[0]{\ensuremath{\oplus}}

\newcommand{\execTS}[1]{\ensuremath{[[ #1 ]]}}

\newcommand{\ite}[3]{\textit{ite}(#1,#2,#3)}
\newcommand{\iteRV}[3]{\mathsf{ite}(#1,#2,#3)}
\newcommand{\eqd}[0]{\ensuremath{=_{\textit{\small{def}}}}}
\newcommand{\instruction}[1]{\xspace{\sf #1}\xspace}
\newcommand{\ifB}[0]{\instruction{if~}}
\newcommand{\thenB}[0]{\instruction{~then~}}
\newcommand{\elseB}[0]{\instruction{~else~}}
\newcommand{\skipB}[0]{\instruction{skip}}

\newcommand{\select}[2]{\ensuremath{\mathsf{rd}(#1,#2)}}
\newcommand{\store}[3]{\ensuremath{\mathsf{wr}(#1,#2,#3)}}
\newcommand{\InsO}[0]{\ensuremath{\mathsf{ins}}}
\newcommand{\Ins}[2]{\ensuremath{\InsO(#1,#2)}}
\newcommand{\enumO}[0]{\ensuremath{\mathsf{enum}}}
\newcommand{\enum}[1]{\ensuremath{\enumO(#1)}}

\newcommand{\PRE}[1]{\ensuremath{\langle #1 \rangle}}
\newcommand{\PREB}[1]{\ensuremath{[ #1 ]}}
\newcommand{\Pre}[2]{\ensuremath{\textit{Pre}_{#1}(#2)}}
\newcommand{\pretilde}[0]{\ensuremath{\widetilde{\textit{pre}}}}
\newcommand{\post}[0]{\ensuremath{\textit{post}}}
\newcommand{\pretildeR}[0]{\ensuremath{\widetilde{\textrm{pre}}}}
\newcommand{\postR}[0]{\ensuremath{\textrm{post}}}
\newcommand{\Succ}[0]{\ensuremath{\textit{successeur}}}
\newcommand{\Men}[0]{\ensuremath{\textit{meneur}}}
\newcommand{\Gpre}[2]{\ensuremath{G_{#1}(#2)}}
\newcommand{\Fpre}[2]{\ensuremath{F_{#1}(#2)}}
\newcommand{\POST}[1]{\ensuremath{[ #1 ]^{o}}}
\newcommand{\sem}[1]{\ensuremath{\textit{rel}_{#1}}}
\newcommand{\quantToSubst}[0]{\ensuremath{\textit{q2s}}}

\newcommand{\wpr}[0]{\ensuremath{\textit{wp}}\xspace}
\newcommand{\wlp}[0]{\ensuremath{\textit{wlp}}\xspace}
\newcommand{\weaklp}[2]{\ensuremath{\textit{wlp}_{#1}(#2)}\xspace}
\newcommand{\weakp}[2]{\ensuremath{\textit{wp}_{#1}(#2)}\xspace}
\newcommand{\prd}[2]{\ensuremath{\textit{prd}_{#1}({#2})}\xspace}

\newcommand{\myInput}[1]{\begin{center}\begin{tabular}{l}\input{#1}\end{tabular}\end{center}}

\newcommand{\Init}[0]{\ensuremath{\textit{Init}}\xspace}
\newcommand{\Target}[0]{\ensuremath{\textit{Target}}\xspace}

\newcommand{\substs}[0]{\ensuremath{\textit{Substs}}}
\newcommand{\choice}[0]{\ensuremath{[]}}
\newcommand{\paral}[0]{\ensuremath{\mid\mid}}
\newcommand{\card}[1]{\ensuremath{\textit{card}(#1)}}

\newcommand{\fs}[1]{\ensuremath{\mathsf{#1}}}
\newcommand{\sort}[1]{\mbox{\textsc{#1}}}
\newcommand{\dropE}[0]{\ensuremath{\mathsf{dropExistential}}}
\newcommand{\renameF}[0]{\ensuremath{\mathsf{renameFormula}}}
\newcommand{\ppnf}[0]{\ensuremath{\mathsf{de}}}
\newcommand{\apnx}[0]{\ensuremath{\mathsf{mini}}}
\newcommand{\apnxq}[0]{\ensuremath{\mathsf{m}_q}}
\newcommand{\sqs}[0]{\ensuremath{\mathsf{rf}}}
\newcommand{\theoryT}[0]{\ensuremath{{\cal T}}}
\newcommand{\theoryU}[0]{\ensuremath{{\cal T}}}
\newcommand{\vars}[0]{\ensuremath{\mathcal{V}}}

\newcommand{\latice}[0]{\ensuremath{\textit{exp}^{A}}}
\newcommand{\laticeB}[0]{\ensuremath{\latice(B_1,\ldots,B_l)}}
\newcommand{\concretisation}[0]{\ensuremath{\gamma}}
\newcommand{\abstraction}[0]{\ensuremath{\alpha}}

\newcommand{\sfunc}[0]{\ensuremath{\textit{SFUNC}}}
\newcommand{\override}[2]{\ensuremath{#1 \bdres #2 }}
\newcommand{\Partition}[1]{\ensuremath{\mathbbm{P}}(#1)}
\newcommand{\fc}[1]{\ensuremath{\widetilde{#1}}} 
\newcommand{\nul}[0]{\ensuremath{\mathsf{null}}} 
\newcommand{\overR}[3]{\ensuremath {\mathsf{over}(#1,#2,#3)}}
\newcommand{\image}[2]{\ensuremath {\sf{image}(#1,#2)}}


\newcommand{\ind}[0]{\textit{index}}
\newcommand{\val}[0]{\textit{val}}
\newcommand{\arr}[0]{\textit{array}}
\newcommand{\iSet}[0]{\ensuremath{\textit{S}_{i}}\xspace}
\newcommand{\eSet}[0]{\ensuremath{\textit{S}_{e}}\xspace}
\newcommand{\ARR}[0]{\textit{ARRAY}\xspace}
\newcommand{\IND}[0]{\textit{INDEX}\xspace}
\newcommand{\VAL}[0]{\textit{VALUE}\xspace}
\newcommand{\constArr}{\ensuremath{\mathsf{const}}\xspace}
\newcommand{\im}[0]{\ensuremath{\mathsf{im}}}

\newcommand{\eJ}{\EuScript J}
\newcommand{\eF}{\EuScript F}
\newcommand{\eL}{\EuScript L}
\newcommand{\eT}{\EuScript T}
\newcommand{\eV}{\EuScript V}
\newcommand{\eP}{\EuScript P}
\renewcommand{\le}{\leqslant}
\renewcommand{\ge}{\geqslant}

\newcommand{\any}[0]{\textrm{@}}

\newcommand{\str}[0]{\textit{Strength}}

\newcommand{\sortSet}{\mathcal{S}}
\newcommand{\sortVars}{\mathcal{X}_s}
\newcommand{\sortFunctionSet}{\mathcal{F}_s} 
\newcommand{\varSet}{\mathcal{X}}
\newcommand{\functionSet}{\mathcal{F}} 

\def\unique{\sort{u}}

\title{Steganography: a Class of Algorithms having Secure Properties}

\author{Jacques M. Bahi, Jean-Fran\c{c}ois Couchot, and Christophe Guyeux*\thanks{* Authors in alphabetic order}\\
University of Franche-Comté, Computer Science Laboratory, Belfort, France\\
Email: \{jacques.bahi, jean-francois.couchot, christophe.guyeux\}@univ-fcomte.fr}

\maketitle

\begin{abstract}
Chaos-based approaches are frequently proposed
in information hiding, but without obvious justification. 
Indeed, the reason why chaos is useful to tackle with discretion, 
robustness, or security, is rarely elucidated. 
This research work presents a new class of non-blind information hiding
algorithms based on some finite domains iterations that are Devaney's 
topologically chaotic.
The approach is entirely formalized and reasons to take place into the
mathematical theory of chaos are explained.
Finally, stego-security and chaos security are consequently proven
for a large class of algorithms.  


\end{abstract}


\section{Introduction}\label{sec:intro}
Chaos-based approaches are frequently proposed to improve
the quality of schemes in information 
hiding~\cite{Wu2007,Liu07,CongJQZ06,Zhu06}.
In these works, the understanding of chaotic systems
is almost intuitive: a kind of noise-like spread system
with sensitive dependence on initial condition.
Practically, some well-known chaotic maps are used
either in the data encryption stage~\cite{Liu07,CongJQZ06}, 
in the embedding into the carrier medium,
or in both~\cite{Wu2007,Wu2007bis}.

This work focus on non-blind  binary information hiding scheme: 
the  original host  is  required  to  extract the  binary hidden
information. This  context is indeed not
as restrictive as it could primarily appear.
Firstly, it allows to prove 
the authenticity of a document sent  through the  Internet 
(the original document is stored whereas the stego content is sent). 
Secondly, Alice and Bob can establish an hidden channel into a
streaming video 
(Alice and Bob both have the same movie, and  Alice hide information into 
the frame number $k$  iff the binary digit
number $k$ of its hidden message is  1).
Thirdly, based on a similar idea, a same
given image can be marked  several times by using various secret parameters
owned both by Alice and Bob. Thus more than one bit can be embedded into a given
image  by using  dhCI  dissimulation. Lastly,  non-blind watermarking  is
useful in  network's anonymity and intrusion  detection \cite{Houmansadr09}, and
to protect digital data sending through the Internet \cite{P1150442004}.

Methods referenced above are almost based on two
fundamental chaotic maps, namely the Chebychev and logistic maps, which range in $\mathbb{R}$. 
To avoid justifying that functions which are chaotic in $\mathbb{R}$
still remain chaotic in the computing representation (\textit{i.e.},
floating numbers) we argue that functions should be iterated on finite domains.
Boolean discrete-time dynamical systems (BS) are thus iterated.

Furthermore, previously referenced works often focus on  
discretion and/or robustness properties, but they do not consider security.  
As far as we know, stego-security~\cite{Cayre2008} and chaos-security 
have only been proven 
on the spread spectrum watermarking~\cite{Cox97securespread},
and on
the dhCI algorithm~\cite{gfb10:ip}, which is notably based on iterating  
the negation function.
We argue that other functions can provide algorithms as secure as the dhCI one.
This work generalizes thus this latter algorithm and formalizes all
its stages. Due to this formalization, we address the proofs of the two
security properties for a large class of steganography approaches.

This research work is organized as follows.
Section~\ref{sec:bs} first recalls the BS context.
The new class of algorithms, 
which is the first contribution,
is firstly introduced in Sec.~\ref{sec:formalization}.
Section~\ref{sec:security} shows how secure  
is our approach: this is the second contribution of the present paper.
Instances of algorithms guaranteeing that desired properties
are presented in Sec.~\ref{instantiating}.
Discussion, conclusive remarks, and perspectives are given in the 
final section.

\section{Boolean Discrete-Time Dynamical Systems}\label{sec:bs}
In this section, we first give some recalls on Boolean discrete
dynamical Systems (BS). 
With this material, next sections 
formalize the information hiding algorithms based on chaotic iterations.

Let $n$ be a positive integer. A Boolean discrete-time 
network is a discrete dynamical
system defined from a {\emph{Boolean map}}
$f:\Bool^n\to\Bool^n$ s.t. 
\[
  x=(x_1,\dots,x_n)\mapsto f(x)=(f_1(x),\dots,f_n(x)),
\]
and an {\emph{iteration scheme}} (\textit{e.g.}, parallel, serial,
asynchronous\ldots). 
With the parallel iteration scheme, 
the dynamics of the system are described by $x^{t+1}=f(x^t)$
where $x^0 \in \Bool^n$.
Let thus $F_f: \llbracket1;n\rrbracket\times \Bool^{n}$ to $\Bool^n$ 
be defined by
\[
F_f(i,x)=(x_1,\dots,x_{i-1},f_i(x),x_{i+1},\dots,x_n),
\]
with the \emph{asynchronous} scheme,
the dynamics of the system are described by $x^{t+1}=F_f(s_t,x^t)$
where $x^0\in\Bool^n$ and $s$ is a {\emph{strategy}}, \textit{i.e.}, a sequence 
in $\llbracket1;n\rrbracket^\Nats$.
Notice that this scheme only modifies one element at each iteration.

Let $G_f$ be the map from $\llbracket1;n\rrbracket^\Nats\times\Bool^n$ to 
itself s.t.
\[
G_f(s,x)=(\sigma(s),F_f(s_0,x)),
\] 
where $\sigma(s)_t=s_{t+1}$ for all $t$ in $\Nats$. 
Notice that parallel iteration of $G_f$ from an initial point
$X^0=(s,x^0)$ describes the ``same dynamics'' as the asynchronous
iteration of $f$ induced by the initial point $x^0$ and the strategy
$s$.

Finally, let $f$ be a map from $\Bool^n$ to itself. The
{\emph{asynchronous iteration graph}} associated with $f$ is the
directed graph $\Gamma(f)$ defined by: the set of vertices is
$\Bool^n$; for all $x\in\Bool^n$ and $i\in \llbracket1;n\rrbracket$,
$\Gamma(f)$ contains an arc from $x$ to $F_f(i,x)$.

\section{Formalization of Steganographic Methods}
\label{sec:formalization}


The data hiding scheme presented here does not constrain media to have 
a constant size. It is indeed sufficient to provide a function and a strategy 
that may be  parametrized with the size of the elements to modify. 
The \emph{mode} and the \emph{strategy-adapter} defined below achieve 
this goal.  

\begin{definition}[Mode]
\label{def:mode}
A map $f$, which associates to any $n \in \mathds{N}$ an application 
$f_n : \mathds{B}^n \rightarrow \mathds{B}^n$, is called a \emph{mode}.
\end{definition}

For instance, the \emph{negation mode} is defined by the map that
assigns to every integer $n \in \mathds{N}^*$ the function 
${\neg}_n:\mathds{B}^n \to \mathds{B}^n, 
{\neg}_n(x_1, \hdots, x_n) \mapsto (\overline{x_1}, \hdots, \overline{x_n})$.

\begin{definition}[Strategy-Adapter]
  \label{def:strategy-adapter}
  A \emph{strategy-adapter}\index{configuration} is a function $\mathcal{S}$, 
  from $\Nats$ to the set of integer sequences, 
  that associates to $n$ a sequence 
  $S \in  \llbracket 1, n\rrbracket^\mathds{N}$.
\end{definition}

Intuitively, a strategy-adapter aims at generating a strategy 
$(S^t)^{t \in \Nats}$ where each term $S^t$ belongs to 
$\llbracket 1, n \rrbracket$.



Let us notice that the terms of $x$ that may be replaced by terms issued
from $y$ are less important than other: they could be changed 
without be perceived as such. More generally, a 
\emph{signification function} 
attaches a weight to each term defining a digital media,
w.r.t. its position $t$:

\begin{definition}[Signification function]
A \emph{signification function} is a real sequence 
$(u^k)^{k \in \Nats}$. 
\end{definition}



For instance, let us consider a set of    
grayscale images stored into 8 bits gray levels.
In that context, we consider 
$u^k = 8 - (k  \mod  8)$  to be the $k$-th term of a signification function 
$(u^k)^{k \in \Nats}$. 

\begin{definition}[Significance of coefficients]\label{def:msc,lsc}
Let $(u^k)^{k \in \Nats}$ be a signification function, 
$m$ and $M$ be two reals s.t. $m < M$. Then 
the \emph{most significant coefficients (MSCs)} of $x$ is the finite 
  vector $u_M$, 
the \emph{least significant coefficients (LSCs)} of $x$ is the 
finite vector $u_m$, and 
the \emph{passive coefficients} of $x$ is the finite vector $u_p$ such that:
\begin{eqnarray*}
  u_M &=& \left( k ~ \big|~ k \in \mathds{N} \textrm{ and } u^k 
    \geqslant M \textrm{ and }  k \le \mid x \mid \right) \\
  u_m &=& \left( k ~ \big|~ k \in \mathds{N} \textrm{ and } u^k 
  \le m \textrm{ and }  k \le \mid x \mid \right) \\
   u_p &=& \left( k ~ \big|~ k \in \mathds{N} \textrm{ and } 
u^k \in ]m;M[ \textrm{ and }  k \le \mid x \mid \right)
\end{eqnarray*}
 \end{definition}

For a given host content $x$,
MSCs are then ranks of $x$  that describe the relevant part
of the image, whereas LSCs translate its less significant parts.
We are then ready to decompose an host $x$ into its coefficients and 
then to recompose it. Next definitions formalize these two steps. 

\begin{definition}[Decomposition function]
Let $(u^k)^{k \in \Nats}$ be a signification function, 
$\mathfrak{B}$ the set of finite binary sequences,
$\mathfrak{N}$ the set of finite integer sequences, 
$m$ and $M$ be two reals s.t. $m < M$.  
Any host $x$ may be decomposed into 
\[
(u_M,u_m,u_p,\phi_{M},\phi_{m},\phi_{p})
\in
\mathfrak{N} \times 
\mathfrak{N} \times 
\mathfrak{N} \times 
\mathfrak{B} \times 
\mathfrak{B} \times 
\mathfrak{B} 
\]
where
\begin{itemize}
\item $u_M$, $u_m$, and $u_p$ are coefficients defined in Definition  
\ref{def:msc,lsc};
\item $\phi_{M} = \left( x^{u^1_M}, x^{u^2_M}, \ldots,x^{u^{|u_M|}_M}\right)$;
 \item $\phi_{m} = \left( x^{u^1_m}, x^{u^2_m}, \ldots,x^{u^{|u_m|}_m} \right)$;
 \item $\phi_{p} =\left( x^{u^1_p}, x^{u^2_p}, \ldots,x^{u^{|u_p|}_p}\right) $.
 \end{itemize}
The function that associates the decomposed host to any digital host is 
the \emph{decomposition function}. It is 
further referred as $\textit{dec}(u,m,M)$ since it is parametrized by 
$u$, $m$ and $M$. Notice that $u$ is a shortcut for $(u^k)^{k \in \Nats}$.
\end{definition}

\begin{definition}[Recomposition]
Let 
$(u_M,u_m,u_p,\phi_{M},\phi_{m},\phi_{p}) \in 
\mathfrak{N} \times 
\mathfrak{N} \times 
\mathfrak{N} \times 
\mathfrak{B} \times 
\mathfrak{B} \times 
\mathfrak{B} 
$ s.t.
\begin{itemize}
\item the sets of elements in $u_M$, elements in $u_m$, and 
elements in $u_p$ are a partition of $\llbracket 1, n\rrbracket$;
\item $|u_M| = |\varphi_M|$, $|u_m| = |\varphi_m|$, and $|u_p| = |\varphi_p|$.  
\end{itemize}
One may associate the vector 
\[
x = 
\sum_{i=1}^{|u_M|} \varphi^i_M . e_{{u^i_M}} +  
\sum_{i=1}^{|u_m|} \varphi^i_m .e_{{u^i_m}} +  
\sum_{i=1}^{|u_p|} \varphi^i_p. e_{{u^i_p}} 
\]
\noindent where 
$(e_i)_{i \in \mathds{N}}$ is the usual basis of the $\mathds{R}-$vectorial space $\left(\mathds{R}^\mathds{N}, +, .\right)$.
The function that associates $x$ to any 
$(u_M,u_m,u_p,\phi_{M},\phi_{m},\phi_{p})$ following the above constraints 
is called the \emph{recomposition function}.
\end{definition}

The embedding consists in the replacement of the values of 
$\phi_{m}$ of $x$'s LSCs  by $y$. 
It then composes the two decomposition and
recomposition functions seen previously. More formally:

\begin{definition}[Embedding media]
Let $\textit{dec}(u,m,M)$ be a decomposition function,
$x$ be a host content,
$(u_M,u_m,u_p,\phi_{M},\phi_{m},\phi_{p})$ be its image by $\textit{dec}(u,m,M)$, 
and $y$ be a digital media of size $|u_m|$.
The digital media $z$ resulting on the embedding of $y$ into $x$ is 
%
the image of $(u_M,u_m,u_p,\phi_{M},y,\phi_{p})$
by  the recomposition function $\textit{rec}$.
\end{definition}

Let us then define the \emph{dhCI} information hiding scheme:

\begin{definition}[Data hiding dhCI]
 \label{def:dhCI}
Let $\textit{dec}(u,m,M)$ be a decomposition function,
$f$ be a mode, 
$\mathcal{S}$ be a strategy adapter,
$x$ be an host content,\linebreak
$(u_M,u_m,u_p,\phi_{M},\phi_{m},\phi_{p})$ 
be its image by $\textit{dec}(u,m,M)$,
$q$ be a positive natural number,  
and $y$ be a digital media of size $l=|u_m|$.

The \emph{dhCI dissimulation}  maps any
$(x,y)$  to the digital media $z$ resulting on the embedding of
$\hat{y}$ into $x$, s.t.

\begin{itemize}
\item We instantiate the mode $f$ with parameter $l=|u_m|$, leading to 
  the function $f_{l}:\Bool^{l} \rightarrow \Bool^{l}$.
\item We instantiate the strategy adapter $\mathcal{S}$ 
with parameter $y$ (and some other ones eventually). 
This instantiation leads to the strategy $S_y \in \llbracket 1;l\rrbracket ^{\Nats}$.

\item We iterate $G_{f_l}$ with initial configuration $(S_y,\phi_{m})$.
\item $\hat{y}$  is the $q$-th term.
\end{itemize}
\end{definition}

To summarize, iterations are realized on the LSCs of the
host content
(the mode gives the iterate function,  
the strategy-adapter gives its strategy), 
and the last computed configuration is re-injected into the host content, 
in place of the former LSCs.


\begin{figure}[ht]
\centering
\includegraphics[width=8.5cm]{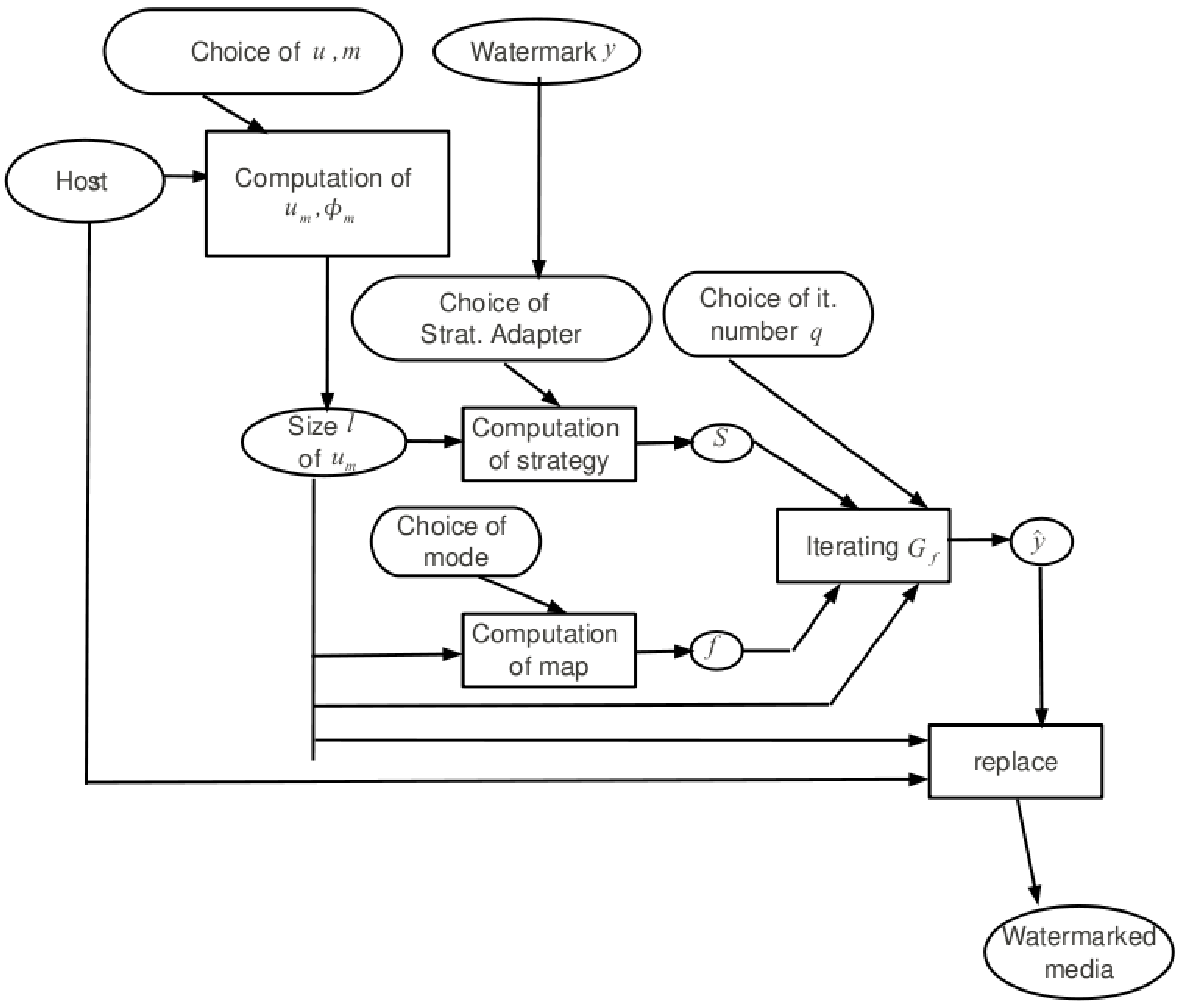}
\vspace{-3em}
\caption{The dhCI dissimulation scheme}
\label{fig:organigramme}
\end{figure}



We are then left to show how to formally check
whether a given digital media $z$ 
results from the dissimulation of $y$ into the digital media $x$. 

\begin{definition}[Marked content]
Let $\textit{dec}(u,m,M)$ be a decomposition function,
$f$ be a mode, 
$\mathcal{S}$ be a strategy adapter, 
$q$ be a positive natural number,  
and  
$y$ be a digital media, 
$(u_M,u_m,u_p,\phi_{M},\phi_{m},\phi_{p})$ be the 
image by $\textit{dec}(u,m,M)$  of  a digital media $x$. 
Then $z$ is \emph{marked} with $y$ if
the image by $\textit{dec}(u,m,M)$ of $z$ is 
$(u_M,u_m,u_p,\phi_{M},\hat{y},\phi_{p})$ where 
$\hat{y}$ is the right member of $G_{f_l}^q(S_y,\phi_{m})$.
\end{definition}

Various decision strategies are obviously  possible to determine whether a given
image $z$ is  marked or not, depending  on the eventuality
that the considered image may have  been attacked.
For example, a  similarity percentage between $x$
and  $z$ can  be  computed, and  the result  can  be compared  to a  given
threshold. Other  possibilities are the use of  ROC curves or
the definition of a null hypothesis problem.
The next section recalls some security properties and shows how the 
\emph{dhCI dissimulation} algorithm verifies them.

\section{Security Analysis}\label{sec:security}


Stego-security~\cite{Cayre2008} is  the  highest security
class in Watermark-Only  Attack setup.
Let $\mathds{K}$ be the set of embedding keys, $p(X)$ the probabilistic model of
$N_0$ initial  host contents,  and $p(Y|K)$ the  probabilistic model  of $N_0$
marked contents s.t. each host  content has  been marked
with the same key $K$ and the same embedding function.

\begin{definition}[Stego-Security~\cite{Cayre2008}]
\label{Def:Stego-security}  The embedding  function  is \emph{stego-secure}
if  $\forall K \in \mathds{K}, p(Y|K)=p(X)$ is established.
\end{definition}


Let us prove that,
\begin{theorem}\label{th:stego}
Let $\epsilon$ be positive,
$l$ be any size of LSCs, 
$X   \sim \mathbf{U}\left(\mathbb{B}^l\right)$,
$f_l$ be an image mode s.t. 
$\Gamma(f_l)$ is strongly connected and 
the Markov matrix associated to $f_l$ 
is doubly stochastic. 
In the instantiated \emph{dhCI dissimulation} algorithm 
with any uniformly distributed (u.d.) strategy-adapter 
which is independent from $X$,  
there exists some positive natural number $q$ s.t.
$|p(X^q)- p(X)| < \epsilon$. 
\end{theorem}

\begin{proof}   
Let $\textit{deci}$ be the bijection between $\Bool^{l}$ and 
$\llbracket 0, 2^l-1 \rrbracket$ that associates the decimal value
of any  binary number in $\Bool^{l}$.
The probability $p(X^t) = (p(X^t= e_0),\dots,p(X^t= e_{2^l-1}))$ for $e_j \in \Bool^{l}$ is thus equal to 
$(p(\textit{deci}(X^t)= 0,\dots,p(\textit{deci}(X^t)= 2^l-1))$ further denoted by $\pi^t$.
Let $i \in \llbracket 0, 2^l -1 \rrbracket$, 
the probability $p(\textit{deci}(X^{t+1})= i)$  is 
\[
 \sum\limits^{2^l-1}_{j=0}  
\sum\limits^{l}_{k=1} 
p(\textit{deci}(X^{t}) = j , S^t = k , i =_k j , f_k(j) = i_k ) 
\]
\noindent 
where $ i =_k j $ is true iff the binary representations of 
$i$ and $j$ may only differ for the  $k$-th element,
and where 
$i_k$ abusively denotes the $k$-th element of the binary representation of 
$i$.

Next, due to the proposition's hypotheses on the strategy,
$p(\textit{deci}(X^t) = j , S^t = k , i =_k j, f_k(j) = i_k )$ is equal to  
$\frac{1}{l}.p(\textit{deci}(X^t) = j ,  i =_k j, f_k(j) = i_k)$.
Finally, since $i =_k j$ are $f_k(j) = i_k$ are constant during the 
iterative process  and thus does not depend on $X^t$, we have 
\[
\pi^{t+1}_i = \sum\limits^{2^l-1}_{j=0}
\pi^t_j.\frac{1}{l}  
\sum\limits^{l}_{k=1} 
p(i =_k j, f_k(j) = i_k ).
\]

Since 
$\frac{1}{l}  
\sum\limits^{l}_{k=1} 
p(i =_k j, f_k(j) = i_k ) 
$ is equal to $M_{ji}$ where  $M$ is the Markov matrix associated to
 $f_l$ we thus have
\[
\pi^{t+1}_i = \sum\limits^{2^l-1}_{j=0}
\pi^t_j. M_{ji} \textrm{ and thus }
\pi^{t+1} = \pi^{t} M.
\]


First of all, 
since the graph $\Gamma(f)$ is strongly connected,
then for all vertices $i$ and $j$, a path can
be  found to  reach $j$  from $i$  in at  most $2^l$  steps.  
There  exists thus $k_{ij} \in \llbracket 1,  2^l \rrbracket$ s.t.
${M}_{ij}^{k_{ij}}>0$.  
As all the multiples $l \times k_{ij}$ of $k_{ij}$ are such that 
${M}_{ij}^{l\times  k_{ij}}>0$, 
we can  conclude that, if
$k$ is the least common multiple of $\{k_{ij}  \big/ i,j  \in \llbracket 1,  2^l \rrbracket  \}$ thus 
$\forall i,j  \in \llbracket  1, 2^l \rrbracket,  {M}_{ij}^{k}>0$ and thus 
$M$ is a regular stochastic matrix.

Let us now recall the following stochastic matrix theorem:
\begin{theorem}[Stochastic Matrix]
  If $M$ is a regular stochastic matrix, then $M$ 
  has an unique stationary  probability vector $\pi$. Moreover, 
  if $\pi^0$ is any initial probability vector and 
  $\pi^{t+1} = \pi^t.M $ for $t = 0, 1,\dots$ then the Markov chain $\pi^t$
  converges to $\pi$ as $t$ tends to infinity.
\end{theorem}

Thanks to this theorem, $M$ 
has an unique stationary  probability vector $\pi$. 
By hypothesis, since $M$ is doubly stochastic we have 
$(\frac{1}{2^l},\dots,\frac{1}{2^l}) = (\frac{1}{2^l},\dots,\frac{1}{2^l})M$
and thus $\pi =  (\frac{1}{2^l},\dots,\frac{1}{2^l})$.
Due to the matrix theorem, there exists some 
$q$ s.t. 
$|\pi^q- \pi| < \epsilon$
and the proof is established.
 \end{proof}
Since $p(Y| K)$ is $p(X^q)$ the method is then stego-secure.

Let us focus now on chaos-security properties.
An information hiding scheme $S$ is said to have such a property
if its iterative process has a chaotic behavior, 
as defined by Devaney, on this topological space.
This problem has been reduced in~\cite{GuyeuxThese10} 
which provides the  following theorem.

\begin{theorem}
\label{Th:Caracterisation des  IC chaotiques} Functions $f  : \mathds{B}^{n} \to
\mathds{B}^{n}$ such that  $G_f$ is chaotic according to  Devaney, are functions
such that the graph $\Gamma(f)$ is strongly connected.
\end{theorem}
\noindent We immediatly deduce:
\begin{corollary}
All the \emph{dhCI dissimulation} algorithms following hypotheses of 
theorem~\ref{th:stego} are chaos-secure.
\end{corollary}

\section{Instantiation of  Steganographic Methods}
Theorem~\ref{th:stego} relies on a u.d.   
strategy-adapter that is independent from the cover, and on an image mode 
$f_l$ whose iteration graph  $\Gamma(f_l)$ is strongly
connected and whose  Markov matrix 
is doubly stochastic.

The CIIS strategy adapter~\cite{gfb10:ip} has the required 
properties: it does not depend on the cover,
and the proof that its outputs are u.d. 
on $\llbracket 1, l \rrbracket$ 
is left as an exercise for the reader
(a u.d. repartition is generated by the  piecewise linear chaotic maps 
and is preserved by the iterative process).
Finally, \cite{wcbg11:ip} has presented an iterative approach to 
generate image modes $f_l$ such that  
$\Gamma(f_l)$ is strongly connected.
Among these maps, it is obvious to check which verifies or not
the doubly stochastic constrain.

\label{instantiating}






\section{Conclusion}\label{sec:concl}
This work has presented a new class of information hiding
algorithms which generalizes 
algorithm~\cite{gfb10:ip} reduced to the negation mode.
Its complete formalization has allowed to prove the stego-security 
and chaos security properties.
As far as we know, this is the first time a whole class of algorithm 
has been proven to have these two properties. 

In future work, our intention is to study the robustness of this class of 
dhCI dissimulation schemes.
We are to find the optimized parameters (modes, stretegy adapters, signification coefficients, iterations numbers\ldots) giving 
the strongest robustness 
(depending on the chosen representation domain), theoretically and practically
by realizing comprehensive simulations.
Finally these algorithms will be compared to other existing ones, among other
things by regarding whether these algorithms are chaotic or not.

\bibliographystyle{IEEEtran}
\bibliography{abbrev,mabase,biblioand}

\begin{thebibliography}{10}
\providecommand{\url}[1]{#1}
\csname url@samestyle\endcsname
\providecommand{\newblock}{\relax}
\providecommand{\bibinfo}[2]{#2}
\providecommand{\BIBentrySTDinterwordspacing}{\spaceskip=0pt\relax}
\providecommand{\BIBentryALTinterwordstretchfactor}{4}
\providecommand{\BIBentryALTinterwordspacing}{\spaceskip=\fontdimen2\font plus
\BIBentryALTinterwordstretchfactor\fontdimen3\font minus
  \fontdimen4\font\relax}
\providecommand{\BIBforeignlanguage}[2]{{%
\expandafter\ifx\csname l@#1\endcsname\relax
\typeout{** WARNING: IEEEtran.bst: No hyphenation pattern has been}%
\typeout{** loaded for the language `#1'. Using the pattern for}%
\typeout{** the default language instead.}%
\else
\language=\csname l@#1\endcsname
\fi
#2}}
\providecommand{\BIBdecl}{\relax}
\BIBdecl

\bibitem{Wu2007}
X.~Wu, Z.-H. Guan, and Z.~Wu, ``A chaos based robust spatial domain
  watermarking algorithm,'' in \emph{ISNN '07: 4th international symposium on
  Neural Networks}, ser. LNCS, vol. 4492.\hskip 1em plus 0.5em minus
  0.4em\relax Springer, 2007, pp. 113--119.

\bibitem{Liu07}
Z.~Liu and L.~Xi, ``Image information hiding encryption using chaotic
  sequence,'' in \emph{KES '07: Knowledge-Based Intelligent Information and
  Engineering Systems}, ser. LNCS, vol. 4693.\hskip 1em plus 0.5em minus
  0.4em\relax Springer, 2007, pp. 202--208.

\bibitem{CongJQZ06}
J.~Cong, Y.~Jiang, Z.~Qu, and Z.~Zhang, ``A wavelet packets watermarking
  algorithm based on chaos encryption,'' in \emph{Computational Science and Its
  Applications ICCSA 2006, International Conference}, ser. LNCS, vol.
  3980.\hskip 1em plus 0.5em minus 0.4em\relax Springer, 2006, pp. 921--928.

\bibitem{Zhu06}
Z.~Congxu, L.~Xuefeng, and L.~Zhihua, ``Chaos-based multipurpose image
  watermarking algorithm,'' \emph{Wuhan University Journal of Natural
  Sciences}, vol.~11, pp. 1675--1678, 2006.

\bibitem{Wu2007bis}
X.~Wu and Z.-H. Guan, ``A novel digital watermark algorithm based on chaotic
  maps,'' \emph{Physics Letters A}, vol. 365, no. 5-6, pp. 403 -- 406, 2007.

\bibitem{Houmansadr09}
A.~Houmansadr, N.~Kiyavash, and N.~Borisov, ``Rainbow: A robust and invisible
  non-blind watermark for network flows,'' in \emph{NDSS’09: 16th Annual
  Network and Distributed System Security Symposium}, 2009.

\bibitem{P1150442004}
G.S.El-Taweel, H.~Onsi, M.Samy, and M.~Darwish, ``Secure and non-blind
  watermarking scheme for color images based on dwt,'' \emph{ICGST
  International Journal on Graphics, Vision and Image Processing}, vol.~05, pp.
  1--5, April 2005.

\bibitem{Cayre2008}
F.~Cayre and P.~Bas, ``Kerckhoffs-based embedding security classes for woa data
  hiding,'' \emph{IEEE Transactions on Information Forensics and Security},
  vol.~3, no.~1, pp. 1--15, 2008.

\bibitem{Cox97securespread}
I.~J. Cox, S.~Member, J.~Kilian, F.~T. Leighton, and T.~Shamoon, ``Secure
  spread spectrum watermarking for multimedia,'' \emph{IEEE Transactions on
  Image Processing}, vol.~6, pp. 1673--1687, 1997.

\bibitem{gfb10:ip}
C.~Guyeux, N.~Friot, and J.~M. Bahi, ``Chaotic iterations versus
  spread-spectrum: Chaos and stego security,'' in \emph{IIH-MSP '10: the 2010
  Sixth International Conf. on Intelligent Information Hiding and Multimedia
  Signal Processing}.\hskip 1em plus 0.5em minus 0.4em\relax IEEE, 2010, pp.
  208--211.

\bibitem{GuyeuxThese10}
C.~Guyeux, ``Le d\'{e}sordre des it\'{e}rations chaotiques et leur utilit\'{e}
  en sécurit\'{e} informatique,'' Ph.D. dissertation, Universit\'{e} de
  Franche-Comt\'{e}, 2010.

\bibitem{wcbg11:ip}
W.~Qianxue, J.~Bahi, J.-F. Couchot, and C.~Guyeux, ``Class of trustworthy
  pseudo-random number generators,'' in \emph{INTERNET'2011. The 3rd Int. Conf.
  on Evolving Internet}, 2011.

\end{thebibliography}

\end{document}